\documentclass[runningheads]{llncs}
\usepackage{amsfonts,amsmath,stmaryrd,bm,hyperref,tikz}
\usetikzlibrary{patterns,arrows,arrows.meta}
\newtheorem{observation}{Observation}
\title{The Price of Universal Temporal Reachability} 
\author{Binh-Minh Bui-Xuan\inst{1}
\and Nhat-Minh Nguyen\inst{1}
\and M\'elissa Setbel\inst{1}
\and S\'ebastien Tixeuil\inst{1}
\and Yukiko Yamauchi\inst{2}}
\authorrunning{B.M. Bui-Xuan, N.M. Nguyen, M. Setbel, S. Tixeuil, and Y. Yamauchi} 
\institute{LIP6, CNRS, Sorbonne Universit\'e, France, \email{[buixuan,nhat-minh.nguyen,melissa.setbel,sebastien.tixeuil]@lip6.fr}
\and Department of Informatics, Kyushu University, Japan, \email{yamauchi@inf.kyushu-u.ac.jp}}

\begin{document}

\maketitle

\begin{abstract}
Dynamic networks are graphs in which edges are available only at specific time instants, modeling connections that change over time.
The dynamic network creation game studies this setting as a strategic interaction where a set of $n$ vertices represents players.
Players can add or remove time-labelled edges in order to minimize their personal cost.
This cost has two components: a construction cost, calculated as the number of time instants during which a player maintains edges multiplied by a constant $\alpha$, and a communication cost, defined as the sum of temporal distances to all other vertices in the network.
Communication occurs through temporal paths, which are sequences of adjacent edges with strictly increasing time labels and no repeated vertices. 
For the strict temporal shortest distance, we show a constant bound for the price of anarchy for $\alpha<2$.
When $\alpha\geq\frac{1}{8}n^2$, we show Nash equilibria with a communication cost asymptotically higher than $\frac{1}{49}n^3$.
We prove that the game admits a best response cycle: there exists a sequence of individually rational strategy changes that returns to the initial strategy profile without ever reaching a Nash equilibrium.
This establishes that the dynamic network creation game is not a potential game and that best response dynamics may cycle indefinitely, so convergence to a Nash equilibrium or a socially optimal outcome is not guaranteed in general.

\end{abstract}

\section{Introduction}
\label{sec:introduction}
Networks such as the Internet are not planned centrally but emerge from the decisions of self-interested participants, who pay for the connections they provide and benefit from being close to everyone else in the network they jointly create. The classical abstraction of this process treats a link as a permanent asset: an agent pays a one-off construction cost for every link it builds and is then charged according to its distances to all other agents in the resulting static network. In many real settings, however, a link is not owned but rented or powered for a limited time: a peer-to-peer connection kept alive during a synchronization window, a satellite or vehicle-to-vehicle radio link that is switched on for a fixed slot, or a leased line billed per time step. An agent who maintains a link for one time step pays for exactly that step, so the natural strategic decision is not only whether to establish a link but when to activate it. This changes the strategic situation fundamentally: a route from one agent to another now exists only if the intermediate links are active in a time-respecting order, so the network's usefulness depends on the coordination of activation times across agents. Moreover, mere temporal reachability is a coarse measure of an agent's utility. Two agents may be reachable from one another, yet a message may need to traverse many links, each of which consumes bandwidth or energy, incurs transmission fees, or adds a risk of failure; in ad hoc and sensor networks every hop drains residual battery, and in communication networks every hop consumes capacity and adds latency. Hop distance along a temporal path therefore captures the cost of using the network once it has been created and complements the per-time-step price of providing it. Our model accounts for both sides of this ledger: agents pay for the temporal usage of the links they activate and for the hop distances they incur when communicating with the other agents.

The \emph{Network creation game} (NCG) was first introduced by Fabrikant et al.~\cite{FLMPS03} and has been extensively investigated as a game-theoretic model to understand networks formed by selfish agents.
An NCG is defined by the strategy and cost at agents in the typical non-cooperative strategic game fashion.
The strategy space of an agent is the subset of agents, that means the agent constructs a communication edge at a cost of $\alpha$ to each agent in its strategy.
A strategy profile (i.e., the tuple of strategies) of all agents defines a communication network, where each communication link can be used by other agents in both directions.
The cost of an agent is the sum of the construction cost, which is the size of its strategy multiplied by $\alpha$, and the communication cost, which is the sum of distance to all the other agents on the communication network formed by the current strategy profile.
Each agent can change its strategy so that its cost is minimized, but selfish behavior of agents may result in a \emph{Nash equilibrium}, where no agent can strictly reduce its cost by changing its own strategy.
The social cost of an NCG is the sum of the cost of all agents, and the quality of a communication network formed by NCG is evaluated by the \emph{Price of Anarchy (PoA)}, which is the ratio between the maximum social cost of a Nash equilibrium and optimal social cost.

Fabrikant et al.\ showed that any tree Nash equilibrium admits the PoA of at most five and asked whether there exists a constant $A$ such that every Nash equilibrium is a tree when $\alpha>A$~\cite{FLMPS03}.
Later, Albers et al.\ introduce the ``Tree Conjecture'', which asks whether every Nash equilibrium is a tree when $\alpha>n$~\cite{AEEMR06}.
The authors proved the tree conjecture for $\alpha > 12 n \lceil \log n \rceil$.
Recently, Dippel and Vetta proved the conjecture for $\alpha > 2n$~\cite{DV24}.
However, Fabrikant et al.\ pointed out that computing a best response of NCG is NP-hard in their seminal paper~\cite{FLMPS03}, and Kawald and Lenzner showed that an NCG may admit a best response cycle~\cite{KL13}.
That is, selfish players cannot always form an efficient communication network by rational strategy changes.
Alon et al.\ introduced a restricted version of the NCG, called the \emph{basic network creation game} (BNCG), where an agent can only remove one incident edge and add another incident edge at one time~\cite{ADML10}.
The BNCG removes the construction cost since any strategy change cannot change the size of a strategy profile, and it focuses on the structure of a communication network.
Additionally, a best response of a swap game can be computed in a polynomial time.
The authors showed that the diameter of any tree Nash equilibrium of the BNCG is two (i.e., a star), while that of a general graph is $2^{O(\sqrt{\lg n})}$.
Lenzner showed that BNCG starting in a tree communication graph is a \emph{potential game}, i.e., there exists a potential function on the strategy profile that monotonically decreases by selfish strategy changes.
They showed that BNCG converges to an equilibrium in $O(n^3)$ strategy changes, while it admits a best response cycle on a general graph.

Inspired by large computer networks, where users and computers are subject to hardware limitations, locality, and asynchrony, NCGs under different constraints have been introduced, some of which exhibit a linear price of anarchy~\cite{bilo24siamdm,moscibroda11algorithmica}.
Bil\`{o} et al.\ introduced an NCG with local information, where each player knows strategies of other players in its $k$-neighbourhood and considers the worst case for the other players in the outside of its visibility~\cite{BGLP14}.
Cord-Landwehr and Lenzner introduced an NCG where the strategy space of a player is determined by other players within distance $k$ but the communication cost of a player depends on the entire communication network~\cite{CL15}. 
Bil\`{o} et al.\ introduced the NCG on a host graph, where agents cannot construct a communication edge not included in the host graph~\cite{BGLP15}.
Bil\`{o} et al.\ also introduced the \emph{generalized network creation game}, where the host graph is a complete edge-weighted graph and the cost of an edge is proportional to the weight of the edge~\cite{bilo24siamdm}.

In recent years, temporal graphs have attracted attention as a fundamental mathematical framework for modeling the dynamic behavior of computer networks.
A temporal graph is an edge-labelled graph, where a label of an edge shows when the edge is available.
Bil\`{o} et al.\ introduced the \emph{temporal reachability network creation game}, where the players buy the edges of a host temporal graph~\cite{BCFGKLS23}.
Specifically, the authors consider a single-labelled temporal graph, where each edge is available only at a particular time instance.
The strategy space of a player is a subset of players, which corresponds to a set of temporal edges, i.e., a tuple of an edge and time instance. 
The construction cost of a player is the size of its strategy, and the communication cost is the number of unreachable players multiplied by a constant $K$, where two players are reachable when there exists a monotonically increasing temporal path. 
They showed that PoA is in $\Omega(\log n)$ for $n$ players by showing that a $\log n$-dimensional hypercube is a Nash equilibrium while a spanning tree is a social optimum.
This result matches the optimum and minimal network design for temporal graphs by Akrida et al.~\cite{Akrida2017}.
The authors showed that given any graph $G$, the minimum number of time labels required for reachability between any pair of vertices is achieved by assigning time labels along a spanning tree of $G$ and it is $2n-3$, while there exists a minimal labelling along a hypercube that requires $\Theta(n \log n)$ labels and removal of any time label loses connectivity of some pair of vertices.
Bil\`{o} et al.\ further extended the game with the terminal model, where the players consider reachability to a common subset of terminal players, and the \emph{global edge buying model}, where players can construct edges between any players~\cite{BCFGKL25}.
They showed that PoA for $k$ terminals is $\Omega(\sqrt{k})$ in the global buying model, while it is $\Omega(\log k)$ for the original local buying model.
Thus, global strategy does not improve the PoA.
These two papers both point out that it is important to consider players who want to minimize the arrival time~\cite{BCFGKLS23} or the total distance~\cite{BCFGKL25} on a temporal graph.

In this paper, we adopt the shortest distance instead of reachability for the communication cost of an NCG on temporal graphs.
A temporal path with strictly increasing time labels is a sequence of temporal edges (i.e., labelled edges), where the label of edges strictly increases.
Throughout the technical part of the paper, temporal paths are strict: their labels must increase strictly, and the distance analysed here is the minimum number of edges in such a strict temporal path.
This differs from temporal-network models that allow non-strict temporal paths, and the results below should be read for the strict temporal shortest-distance model only.
Then, we adopt the original cost function for a player, i.e., the sum of the size of its strategy multiplied by $\alpha$ and the sum of the shortest distance to all the other players.
For dynamic games, we show for $\alpha\geq\frac{1}{8}n^2$ scenarios where the temporal constraints can create communication cost as high as $\frac{1}{48}(n-2)^3$.
Since $n^3$ is a trivial upper bound of the communication cost,
this hints at a very bad aftermath for the dynamic network creation game when players are uncooperative.
However, our current constructions require the extra condition over $n$ and $\alpha$ which do not establish an unconditional unbounded lower bound for the price of anarchy.
We conjecture that the price of anarchy in the dynamic game is lower bounded by an unbounded function of $n$.
On the other hand, we show for the price of stability, the ratio of the best social cost of a Nash equilibrium over the optimal social cost, that it is at most a constant times $\alpha$, when $n\to+\infty$.
Additionally, we prove that for $\alpha \geq 3$ the game admits a best response cycle of length six, involving three players that repeatedly toggle their direct edges to a common target vertex.
This shows that the dynamic network creation game is not a potential game~\cite{KL13}: no real-valued potential function decreases monotonically along every sequence of best responses, so best response dynamics may cycle indefinitely and convergence to a Nash equilibrium is not guaranteed.

Unlike the temporal NCG of Bil\`{o} et al.~\cite{BCFGKLS23}, which restricts players to buying edges from a fixed host temporal graph with predetermined time labels, our model gives players full freedom to choose both the target vertex and the time label of each purchased edge.
This is the natural temporal analog of the original NCG by Fabrikant et al.~\cite{FLMPS03}, and it is the appropriate setting for studying \emph{de novo} temporal network formation by selfish agents.
Contrasting the static network creation game's ``Tree Conjecture'' where trees (and stars) play a crucial role since they are connected graphs with the least number of edges, computing the optimal size of a temporal spanner is an open question~\cite{Casteigts_et_al_about_spanners}.
Tree-like structures are no more natural candidates for the dynamic network creation game as Nash equilibria.
In particular, given two vertices $u$ and $v$ of a temporal graph whose underlying graph is a tree, if there is a strict temporal path from $u$ to $v$, then $u$ is unreachable from $v$.
In contrast to this sparse case, if the players only buy the minimum label $1$ for all purchased edges, the network results in a complete graph.
This strategy profile is a Nash equilibrium for the out-reachability shortest-distance game: every target reachable from a player is reached by a direct temporal edge of length $1$, no strict temporal path can use two consecutive edges both labelled with $1$, adding or relabelling edges cannot reduce any distance below $1$, and omitting an owned direct edge makes the corresponding endpoint unreachable.
On the other hand, players might be compelled to buy edges at random time labels (same cost) in the hope that the network would result in a better social cost.
An example is when $4$ players buying a $C_4$ (cycle over $4$ vertices) with time labels $1$, $3$, $2$, $4$, consecutively.
The time-label choices that emerge at Nash equilibria thus reflect genuine temporal constraints, and the game does not trivially reduce to the static NCG.
Our paper provides the first PoA result for temporal NCGs under the shortest-distance metric.

\section{Dynamic network creation game}
\label{sec:definition}
In this paper, we use standard graph theory vocabulary for undirected simple graphs with possible loops~\cite{BM08,D25,KT06}.
We also use the abusive notation $uv=vu=\{u,v\}$ when it is clear from the context that $u$ and $v$ are vertices.
We briefly recall the main difference when considering distance in temporal graphs \textit{vs}.\ in non-temporal graphs, before focusing on the main topic of this paper which is the network creation game.

A \emph{temporal graph} $\mathcal G$ is a tuple $\mathcal G=(V,E,\lambda)$ where $G=(V,E)$ is a graph and $\lambda: E\to 2^{\mathbb N^*}\setminus\{\emptyset\}$.
For every edge $e\in E$, the elements of set $\lambda(e)$ are called the \emph{time labels} corresponding to instants when edge $e$ is considered active.
If $|\lambda(uv)|=1$ for every $uv\in E$, we say that temporal graph $\mathcal G$ has \emph{simple labelling}.

A (strict) \emph{temporal path} from $u$ to $v$ in $\mathcal G$ is a path from $u$ to $v$ in $G$ together with strictly increasing time labels on the edges appearing in the path, that is, a sequence $(e_1,t_1),(e_2,t_2),\dots,(e_k,t_k)$, where $t_i\in\lambda(e_i)$ for $1\leq i\leq k$, where the edges $e_1,e_2,\dots,e_k$ form a path from $u$ to $v$ in $G$, and where $t_1<t_2<\dots<t_k$.
Because of the time labels, a temporal path might exist from $u$ to $v$ when there is no temporal path from $v$ to $u$ (even though $G=(V,E)$ is non-oriented).
Note that we will not address non-strict temporal paths in this paper.

Three variants of \emph{temporal distance} $d_{\mathcal G}(u,v)$ from $u$ to $v$ in $\mathcal G$ can be defined, as follows.
In any case this distance is infinite when there is no temporal path from $u$ to $v$ in $\mathcal G$, otherwise it is called:
\emph{shortest distance} when defined as the least number of edges in a temporal path from $u$ to $v$;
\emph{foremost distance} when defined as the smallest time label $t_k$ for which there exists a temporal path $(e_1,t_1),(e_2,t_2),\dots,(e_k,t_k)$ from $u$ to $v$;
\emph{fastest distance} when defined as the smallest value $t_k-t_1$ for which there exists a temporal path $(e_1,t_1),(e_2,t_2),\dots,(e_k,t_k)$ from $u$ to $v$.

All results in this paper concern the out-reachability version of the below defined game, using uniquely the strict temporal shortest distance function. 
The in-reachability, foremost distance, and fastest distance variants of the game are defined only to delimit the model family and are left for future work. 

The dynamic network creation game we consider has a strictly positive \emph{atomic cost} $\alpha>0$, and is played by a set of $n$ \emph{players} called $V$.
The \emph{strategy} of player $v\in V$ is a finite set $\bm s(v)\in 2^{V\times\mathbb N^*}$.
A \emph{strategy profile} for players $V$ is a sequence of strategies $\bm s=(\bm s(v))_{v\in V}$.

To simplify notations, when it is clear from the context, we sometimes omit the mention of $\bm s$ in definitions stemming from $\bm s$.
For instance, we now let $\mathcal G=\mathcal G[\bm s]$, $E=E[\bm s]$ and $\lambda=\lambda[\bm s]$ and define
the \emph{communication temporal graph $\mathcal G$ of strategy profile $\bm s$} as a tuple $\mathcal G=(V,E,\lambda)$ where
\[
E=\{uv:\exists t\in\mathbb N^*\textrm{ s.t. } (u,t)\in\bm s(v)\lor(v,t)\in\bm s(u)\}
\]
and $\lambda(uv)=\{t:(u,t)\in\bm s(v)\lor(v,t)\in\bm s(u)\}$ for every $uv\in E$.

The \emph{social cost} of playing $\bm s$ is $c[\bm s]=\sum_{v\in V}c[\bm s](v)$, where $c[\bm s](v)$ denotes the \emph{individual cost} of player $v$.
This individual cost has two components, the first corresponding to the cost for maintaining connection in $\mathcal G$, while the second is related to either the out-reachability of the player towards the rest of $\mathcal G$, or the in-reachability to the player from the rest of $\mathcal G$.
We will see later that the definition of social cost will not change, whichever definition of individual cost in use.
However, we show in Observation~\ref{obs:K_n} situations which are Nash equilibria for one definition but not Nash equilibria for the other definition.
Formally, $c[\bm s](v)=c^+[\bm s](v)$ or $c[\bm s](v)=c^-[\bm s](v)$ where
\[
c^+[\bm s](v)=\alpha\times|\bm s(v)|+\sum_{w\in V}d_{\mathcal G}(v,w)
\textrm{~~~and~~~}
c^-[\bm s](v)=\alpha\times|\bm s(v)|+\sum_{u\in V}d_{\mathcal G}(u,v).
\]
Three variants of the dynamic network creation game can be defined, when $d_{\mathcal G}$ is the shortest distance, the foremost distance, or the fastest distance in $\mathcal G$.

Strategy profile $\bm s$ has \emph{optimal social cost} if its social cost is minimum among all possible strategy profiles with the same atomic cost $\alpha$ and the same number $n$ of players.
Strategy profile $\bm s$ is a \emph{Nash equilibrium} if for every player $v\in V$, there is no other strategy profile $\bm s'$ for $V$ with the same atomic cost and $\bm s'(u)=\bm s(u)$, for every $u\neq v$, such that the individual cost of $v$ when playing $\bm s'$ is strictly less than the individual cost of $v$ when playing $\bm s$.
Equivalently, $\bm s$ is a Nash equilibrium if for every player $v\in V$ and any alternative strategy $\bm s'(v)\in 2^{V\times\mathbb N^*}$, we have $c[\bm s](v)\leq c[(\bm s'(v),\bm s_{-v})](v)$, where $\bm s_{-v}=(\bm s(u))_{u\neq v}$ denotes the strategies of all players other than $v$.

Given $\alpha>0$ and $n\in\mathbb N$, the \emph{price of anarchy} $PoA(\alpha,n)$ is the ratio of the worst (maximum) social cost of a Nash equilibrium over the optimal social cost for strategy profiles with atomic cost $\alpha$ and $n$ players.
The \emph{price of stability} $PoS(\alpha,n)$ is defined similarly with the ratio of the best (minimum) social cost of a Nash equilibrium over the optimal cost.
These are well-defined and finite ratios according to Observation~\ref{obs:K_n} below.
Let $K_n$ be the $n$-vertex complete graph.
\begin{observation}\label{obs:K_n}
For the out-reachability dynamic network creation game with atomic cost $\alpha$ and $n$ players, using strict temporal shortest distance, both following numbers are finite: $PoA(\alpha,n)\geq PoS(\alpha,n)\geq1$.
The complete graph $K_n$ with only time labels $1$ on its edges is the communication temporal graph of a Nash equilibrium.
Besides, the communication temporal graph of any Nash equilibrium must have simple labelling.

For $n\geq 3$ and $\alpha>1$, there is a strategy profile with $K_n$ as communication temporal graph which is not a Nash equilibrium for the in-reachability dynamic network creation game with atomic cost $\alpha$ and $n$ players, using strict temporal shortest distance.
\end{observation}
\begin{proof}
Consider any total order over the set of $n$ players, and a strategy profile where every player buys one edge with time label $1$ to every player strictly higher than itself w.r.t.\ the total order.
This results in the complete graph $K_n$ with only time labels~$1$ as a communication temporal graph.
For each player, every reachable target is reached by a direct edge and has distance $1$, which is already minimum. 
Adding edges, changing labels, or buying duplicate temporal incidences cannot reduce any such distance; deleting an owned direct edge makes the corresponding endpoint unreachable because strict temporal paths cannot use two consecutive edges both with a time label $1$.
Thus no arbitrary unilateral strategy replacement is profitable.
In other words, the strategy profile is a Nash equilibrium.
Besides, the social cost of this strategy profile is finite.
Hence, the optimal social cost for $\alpha$ and $n$ is finite.
In other words, $PoS(\alpha,n)\geq1$ is a well-defined finite number.

Suppose that, for some $\alpha$ and some $n$, a Nash equilibrium exists with infinite social cost.
Then, there exist vertices $u$ and $v$ for which the temporal distance from $u$ to $v$ is infinite.
Let $V_u$ be the subset of every vertex $v'$ such that the temporal distance from $u$ to $v'$ is infinite.
But then adding every $v'\in V_u$, together with any time label, to the strategy of $u$ would strictly decrease the individual cost for $u$ in the new strategy profile, from infinite to finite.
Contradiction.
Hence, the social cost of any Nash equilibrium for $\alpha$ and $n$ is finite.
In other words, $PoA(\alpha,n)$ is also a well-defined finite number.

Consider the communication temporal graph of a Nash equilibrium.
Suppose one of its edges has at least two time labels $t<t'$.
For its individual communication cost, the vertex which bought the edge at $t'$ only uses the edge as departure to another vertex.
Since any temporal path departing at $t'$ can also depart at earlier time $t$ instead, the vertex owning $t'$ can remove this purchase without changing its communication cost while strictly decreasing its construction cost, contradicting the definition of a Nash equilibrium.
Hence, the communication temporal graph of a Nash equilibrium must have simple labelling.

As for in-reachability, the first vertex in the total order can remove its purchased edge to the second vertex while changing the time label of its purchased edge to the last vertex to $2$:
its communication cost would increase by $1$ (the second vertex must now transit by the last vertex before reaching the first vertex), while its construction cost strictly decreases by $\alpha$.
\qed
\end{proof}

Let $\bm s$ be a strategy profile with either an optimal social cost or the fact that $\bm s$ is a Nash equilibrium.
We have just seen that the social cost of $\bm s$ is finite.
We would like now to give an equivalent definition of this social cost. 
Let $\mathcal G=\mathcal G[\bm s]=(V,E,\lambda)$ be the communication temporal graph of $\bm s$.
Since each player's strategy $\bm s(v)$ is a \emph{set} in $2^{V\times\mathbb N^*}$, each pair $(u,t)$ can appear at most once in $\bm s(v)$, so each temporal edge is purchased by at most one player (otherwise one player would not buy that temporal edge, decreasing both its optimal selfish cost and the optimal social cost, contradiction). We may therefore identify the social cost $c[\bm s]$ of strategy $\bm s$ with $c(\mathcal G)$, the \emph{social cost of $\mathcal G$}.
It is defined as follows, where the value does not change whether we consider out-reachability or in-reachability for the individual cost:
\[
c(\mathcal G)=c[\bm s]=\sum_{v\in V}c^+[\bm s](v)=\sum_{v\in V}c^-[\bm s](v)=\alpha\times\sum_{e\in E}|\lambda(e)|+\sum_{u,v\in V}d_{\mathcal G}(u,v).
\]
Since $c(\mathcal G)$ is finite, the temporal distance from any vertex $u$ to any vertex $v$ must be finite.
Temporal graph $\mathcal G$ is then called \emph{temporally connected}.
Its \emph{diameter} is defined as the maximum temporal distance from a vertex to another vertex.

\section{Low price of stability}
\label{sec:shortest_small_alpha}
We aim at proving the following results about the price of stability, which are summarized in Figure~\ref{fig:PoS}.
In the out-reachability dynamic network creation game with atomic cost $\alpha$ and $n$ players, using strict temporal shortest distance for communication cost:
\begin{itemize}
    \item if $\alpha\leq4$, then $PoS(\alpha,n)=1$,
    \item if $n\leq6$, then $PoS(\alpha,n)=1$,
    \item otherwise, $PoS(\alpha,n)\leq\frac{\alpha}{4}+\frac{1}{2}$.
\end{itemize}
Asymptotically,
\begin{itemize}
    \item if $n=o(\sqrt\alpha)$, then $\lim\limits_{\alpha\to+\infty}PoS(\alpha,n)=1$,
    \item if $n=\sqrt{8\alpha}$, then $\lim\limits_{\alpha\to+\infty}PoS(\alpha,\sqrt{8\alpha})\leq5$.
\end{itemize}
\begin{figure}[t!]
\centering
\begin{tikzpicture}[xscale=1.2, yscale=0.4, >=stealth]

    \def\xZero{0}
    \def\xFour{0.49}          
    \def\xIntersection{0.55}
    \def\xIntersect{1.8}     
    \def\xBreakStart{5.8}    
    \def\xBreakEnd{6.2}
    \def\xBreakMid{6.0}      
    \def\xEnd{8.2}
    \def\xMax{8.7}

    \def\yZero{0}
    \def\yTwo{1.5}           
    \def\ySix{1.6}           
    \def\yMax{7.5}           

    \pgfmathsetmacro{\k}{\yTwo / sqrt(\xFour)}

    \path[pattern=north east lines, pattern color=green!60!black]
        (0,0) -- (0,\yMax) -- (\xFour, \yMax) -- (\xFour, \ySix) -- (\xEnd, \ySix)
        -- (\xEnd, 0) -- cycle;
    
    \path[pattern=crosshatch, pattern color=blue!60!black]
        (\xIntersection, \ySix)
        -- plot[domain=\xIntersection:\xBreakStart, samples=100] (\x, {\k * sqrt(\x)})
        -- plot[domain=\xBreakEnd+0.2:\xEnd, samples=100] (\x, {\k * sqrt(\x + 0.8)})
        -- (\xEnd, \ySix) -- cycle;

    \path[pattern=north west lines, pattern color=black!70]
        (\xFour, \ySix) -- (\xFour, \yMax) -- (\xEnd, \yMax) -- (\xEnd, {\k * sqrt(\xEnd + 0.8)})
        -- plot[domain=\xEnd:\xBreakEnd+0.2, samples=100] (\x, {\k * sqrt(\x + 0.8)})
        -- plot[domain=\xBreakStart:\xIntersection, samples=100] (\x, {\k * sqrt(\x)})
        -- (\xFour, \ySix) -- cycle;

    \draw[->, thick] (0,0) -- (\xMax, 0) node[below] {$\alpha$};

    \draw[thick, white, fill=white] (\xBreakStart - 0.1, -0.3) rectangle (\xBreakEnd + 0.1, 0.3);
    \draw[thick] (\xBreakStart, -0.3) -- (\xBreakStart + 0.2, 0.3);
    \draw[thick] (\xBreakStart + 0.15, -0.3) -- (\xBreakStart + 0.35, 0.3);

    \draw[->, thick] (0,0) -- (0, \yMax + 0.8) node[left] {$n$};

    \foreach \x/\label in {\xZero/0, \xFour/4} {
        \draw[thick] (\x, 0.2) -- (\x, -0.2) node[below, yshift=-2pt] {$\label$};
    }
    \draw[thick] (\xBreakEnd+1, -0.2) -- (\xBreakEnd+1, -0.2) node[below, yshift=-2pt] {asymptotic at $+\infty$};

    \foreach \y/\label in {\ySix/6} {
        \draw[thick] (0.08, \y) -- (-0.08, \y) node[left] {$\label$};
    }

    \draw[dashed, thick, gray!80] (\xFour, 0) -- (\xFour, \ySix);
    \draw[very thick, green!80!black] (\xFour, \ySix) -- (\xFour, \yMax);
    \draw[dashed, thick, gray!80] (\xBreakMid, 0) -- (\xBreakMid, \yMax);

    \draw[dashed, thick, gray!80] (\xIntersection, \ySix) -- (\xEnd, \ySix);
    \draw[very thick, green!80!black] (\xFour, \ySix) -- (\xIntersection, \ySix);
    \draw[dashed, thick, gray!80] (0, \ySix) -- (\xFour, \ySix);

    \draw[dashed, thick, gray!80]
        plot[domain=0:\xIntersection, samples=100] (\x, {\k * sqrt(\x)});
    \draw[very thick, blue!80!black]
        plot[domain=\xIntersection:\xBreakStart, samples=100] (\x, {\k * sqrt(\x)});

    \pgfmathsetmacro{\yLeftEnd}{\k * sqrt(\xBreakStart)}
    \pgfmathsetmacro{\yRightStart}{\k * sqrt(\xBreakEnd + 0.8)}

    \draw[very thick, blue!80!black, dashed]
        (\xBreakStart, \yLeftEnd) -- (\xBreakEnd + 0.2, \yRightStart);

    \draw[very thick, blue!80!black]
        plot[domain=\xBreakEnd+0.2:\xEnd, samples=100] (\x, {\k * sqrt(\x + 0.8)})
        node[right] {$n = \sqrt{8\alpha}$};

    \node[fill=white] at (4.36, 6)  {$PoS(\alpha,n)\leq\frac{\alpha}{4}+\frac{1}{2}$};
    \node[fill=white] at (5.6, 2.7)  {if $n=o(\sqrt\alpha)$, then $\lim\limits_{\alpha\to+\infty}PoS(\alpha,n)=1$};
    \node[fill=white] at (4.05, 0.8)  {$PoS(\alpha,n)=1$};

\end{tikzpicture}
\caption{Price of stability: upper bounds and asymptotic limits.
If $(\alpha,n)$ is located inside the dashed green regions, then $PoS(\alpha,n)=1$.
Otherwise,
$PoS(\alpha,n)\leq\frac{\alpha}{4}+\frac{1}{2}$.
If $(\alpha,n)$ is located inside the crosshatched blue regions, then $\lim\limits_{\alpha\to+\infty}PoS(\alpha,f(\alpha))=1$ for any $f(\alpha)=o(\sqrt\alpha)$.
}
\label{fig:PoS}
\end{figure}

Let $K_{2,n-2}$ be the $n$-vertex complete bipartite graph where one party has exactly two vertices.
\begin{theorem}
\label{lemma:basic}
The social optimum of the out-reachability dynamic network creation game with atomic cost $\alpha>0$ and $n$ players, when $n\geq4$, using strict temporal shortest distance, is equal to the social cost of a strategy profile whose communication graph is:
\begin{itemize}
    \item $K_n$ and any simple labelling (that is, each edge active at exactly one time instant), if $0<\alpha\leq2$;
    \item $K_{2,n-2}$ together with some well chosen simple labelling (described explicitly in the proof below), if $\alpha\geq2$.
\end{itemize}
The worst Nash equilibrium is realized by a strategy profile whose communication graph is:
\begin{itemize}
    \item $K_n$ and any simple labelling if $0<\alpha<1$, which also implies $1=PoS(\alpha,n)=PoA(\alpha,n)$;
    \item $K_{2,n-2}$ together with some well chosen simple labelling if $1\leq\alpha<2$, moreover, $1=PoS(\alpha,n)\leq PoA(\alpha,n)\leq\frac{4}{3}$. The upper bound $\frac{4}{3}$ is not tight and only achievable asymptotically for $\alpha=1$ and $n\to\infty$.
\end{itemize}
The latter case is still a Nash equilibrium for $2\leq\alpha\leq4$ and $n>6$, which also implies $PoS(\alpha,n)=1$.
\end{theorem}
\begin{proof}
The proof is inspired from~\cite{FLMPS03}.
Let $\mathcal G=\mathcal G[\bm s]=(V,E,\lambda)$ be the communication temporal graph of $\bm s$, a strategy profile with either an optimal social cost or the fact that $\bm s$ is a Nash equilibrium.
For every pair of vertices $u$ and $v$, either $uv\in E$ and we also have $d_{\mathcal G}(u,v)=d_{\mathcal G}(v,u)=1$, or we have both $d_{\mathcal G}(u,v)\geq 2$ and $d_{\mathcal G}(v,u)\geq 2$.
Hence,
\[
\begin{aligned}
c(\mathcal G)=\alpha\sum_{e\in E}|\lambda(e)|+\sum_{\substack{u,v\in V \\ u\neq v}} d_{\mathcal G}(u,v) &\geq \alpha|E| + 2|E|\times1 + (n(n-1)-2|E|)\times 2\\
&\geq 2n(n-1)+(\alpha-2)|E|,
\end{aligned}
\]
where the bound is tight if $\mathcal G$ has simple labelling and diameter at most $2$.

If $\alpha<2$, the social cost is minimized when in addition to the above, $|E|$ is maximized.
This is realized by $K_n$ together with any simple labelling.
If $\alpha\geq2$, the social cost is minimized when $|E|$ is minimized among $\mathcal G$ with simple labelling and diameter at most $2$.
Bumby proved in~\cite{Bumby1981Telephones} that the minimum number of time labels needed for a temporally connected temporal graph is $2n-4$.
This constraint and both the requirement about the diameter and the single labelling can be realized by $K_{2,n-2}$ together with the following labelling.
Note that throughout this proof, $\lambda$ is chosen as a simple labelling, meaning $|\lambda(e)|=1$ for every edge $e$, so that each edge is active at exactly one time instant and contributes exactly $\alpha$ to the construction cost.
Both vertices of degree $(n-2)$ have their incident edges taking distinct labels from $1$ to $n-2$.
The sum of time labels on both incident edges of any vertex of degree $2$ is $n-1$.
When $n=8$, such a labelling is shown in the middle-left part of Figure~\ref{fig:8vertices}.
\begin{figure}[t]
\centering
\begin{minipage}{0.30\textwidth}
\centering
\begin{tikzpicture}[
    scale=0.5,
    transform shape,
    vertex/.style={circle, draw, fill=black, inner sep=1.5pt},
    edge label/.style={midway, inner sep=1pt}
]

\foreach \i/\name in {1/V1, 2/V2, 3/V3, 4/V4, 5/V5, 6/V6, 7/V7, 8/V8} {
    \node[vertex] (\name) at ({360/8*(\i-1)}:3) {};
}

\foreach \a/\b/\c in {
    V1/V2/1,  V1/V3/1,  V1/V4/1,  V1/V5/1,  V1/V6/1,  V1/V7/1,  V1/V8/1,
    V2/V3/1,  V2/V4/1,  V2/V5/1,  V2/V6/1,  V2/V7/1,  V2/V8/1,
    V3/V4/1,  V3/V5/1,  V3/V6/1,  V3/V7/1,  V3/V8/1,
    V4/V5/1,  V4/V6/1,  V4/V7/1,  V4/V8/1,
    V5/V6/1,  V5/V7/1,  V5/V8/1,
    V6/V7/1,  V6/V8/1,
    V7/V8/1
} {
    \draw (\a) -- node[edge label, above] {\c} (\b);
}

\end{tikzpicture}
\end{minipage}
\hfill
\begin{minipage}{0.3\textwidth}
\centering
\begin{tikzpicture}[
    scale=0.65,
    transform shape,
    vertex/.style={circle, draw, fill=black, inner sep=1.5pt},
    edge label/.style={midway, inner sep=1pt}
]

\node[vertex] (L) at (0,0.5) {};

\node[vertex] (M1) at (2, 3) {};
\node[vertex] (M2) at (2, 2) {};
\node[vertex] (M3) at (2, 1) {};
\node[vertex] (M4) at (2, 0) {};
\node[vertex] (M5) at (2,-1) {};
\node[vertex] (M6) at (2,-2) {};

\node[vertex] (R) at (4,0.5) {};

\draw[{Stealth[length=2mm, width=1.5mm]}-] (L) -- node[edge label, above] {1} (M1);
\draw[{Stealth[length=2mm, width=1.5mm]}-] (L) -- node[edge label, above] {2} (M2);
\draw[{Stealth[length=2mm, width=1.5mm]}-] (L) -- node[edge label, above] {3} (M3);
\draw[{Stealth[length=2mm, width=1.5mm]}-] (L) -- node[edge label, above] {4} (M4);
\draw[{Stealth[length=2mm, width=1.5mm]}-] (L) -- node[edge label, above] {5} (M5);
\draw[-{Stealth[length=2mm, width=1.5mm]}] (L) -- node[edge label, above] {6} (M6);

\draw[{Stealth[length=2mm, width=1.5mm]}-] (M1) -- node[edge label, above] {6} (R);
\draw[-{Stealth[length=2mm, width=1.5mm]}] (M2) -- node[edge label, above] {5} (R);
\draw[-{Stealth[length=2mm, width=1.5mm]}] (M3) -- node[edge label, above] {4} (R);
\draw[-{Stealth[length=2mm, width=1.5mm]}] (M4) -- node[edge label, above] {3} (R);
\draw[-{Stealth[length=2mm, width=1.5mm]}] (M5) -- node[edge label, above] {2} (R);
\draw[-{Stealth[length=2mm, width=1.5mm]}] (M6) -- node[edge label, above] {1} (R);

\end{tikzpicture}
\end{minipage}
\hfill
\begin{minipage}{0.30\textwidth} 
\centering
\begin{tikzpicture}[
    scale=0.7,
    transform shape,
    vertex/.style={circle, draw, fill=black, inner sep=1.5pt},
    >={Stealth[length=2mm, width=1.5mm]}, 
    every edge/.style={draw, ->},          
    edge label/.style={midway, font=\small, inner sep=2pt}
]

    \node[vertex] (M1) at (2,  3) {};
    \node[vertex] (M2) at (2,  2) {};
    \node[vertex] (M3) at (2,  1) {};
    \node[vertex] (M4) at (2,  0) {};
    \node[vertex] (M5) at (2, -1) {};
    \node[vertex] (M6) at (2, -2) {};

    \node[vertex] (R1) at (4,  2) {};
    \node[vertex] (R2) at (4,  1) {};

    \path (R2) edge node[edge label, right] {2} (R1);
    \path (R2) edge node[edge label, above] {2} (M2);
    \path (M1) edge node[edge label, above] {1} (R1);
    \path (M2) edge node[edge label, above] {1} (R1);
    \path (M2) edge node[edge label, left ] {3} (M1);
    
    \path (M2) edge[bend right=45] node[edge label, left] {4} (M4);
    
    \path (M3) edge node[edge label, above] {1} (R2);
    \path (M4) edge node[edge label, above] {1} (R2);
    \path (M4) edge node[edge label, left] {5} (M3);
    
    \path (M4) edge[bend right=45] node[edge label, left] {6} (M6);
    
    \path (M5) edge node[edge label, above] {1} (R2);
    \path (M6) edge node[edge label, above] {1} (R2);
    \path (M6) edge node[edge label, left] {7} (M5);

\end{tikzpicture}
\end{minipage}
\caption{Three temporally connected temporal graphs over $n=8$ vertices, whose union graphs are a clique $K_8$ (leftmost), a biclique $K_{2,6}$ (middle), and a graph with a linear structure (rightmost).
With respect to the shortest distance, they have diameter $1$, $2$, and $\frac{n-2}{2}$, respectively.
The arrows on $K_{2,6}$ (middle) and the rightmost graph represent strategy profiles, where vertices buy the arrows from which they originate.
For a small atomic cost $\alpha=1$, the strategy profile for $K_{2,6}$ is a Nash equilibrium whose social cost is worse than $K_8$.
For a high atomic cost $\alpha=100$ the strategy profile for the rightmost graph is a Nash equilibrium, but not the one for $K_{2,6}$ anymore.}
\label{fig:8vertices}
\end{figure}

As for Nash equilibria, note from Observation~\ref{obs:K_n} that their communication temporal graphs have simple labelling.
Now, if $\alpha<1$, any non-edge $uv\notin E$ would be bought by either $u$ or $v$, together with any time label, in order to decrease its cost by at least $1-\alpha>0$ (from the difference between the previous distance from $u$ to $v$ and the new one, offset by the construction cost of $\alpha$).
Hence, the diameter of $\mathcal G$ is $1$ for this case.
It must then be $K_n$ together with any simple labelling.
Hence, $PoA(\alpha,n)=1$.

Similarly as before, in case of Nash equilibrium, if $\alpha<2$, whenever the distance from $u$ to $v$ is at least $3$, vertex $u$ would buy the non-edge $uv$, together with any time label, in order to decrease its cost by at least $2-\alpha>0$ (from the difference between the previous distance from $u$ to $v$ and the new one, offset by the construction cost of $\alpha$).
Hence, the diameter of $\mathcal G$ is at most $2$ for this case.
For $1\leq\alpha<2$, the worst social cost among the candidates with diameter at most $2$ and simple labelling is realized by $K_{2,n-2}$ together with the above defined labelling; the assumption $\alpha\geq1$ is needed for the labelled biclique profile to be a Nash equilibrium (otherwise it would be the clique).
Hence,
\[
PoA(\alpha,n)
=\frac{2n(n-1)+(\alpha-2)(2n-4)}{2n(n-1)+(\alpha-2)(\frac{n(n-1)}{2})}
=\frac{4n^2+4(\alpha-3)n-8\alpha+16}{(\alpha+2)n^2-(\alpha+2)n}
\leq\frac{4}{3}.
\]

It remains to prove for $2\leq\alpha\leq4$ that $PoS(\alpha,n)=1$ when $n>6$.
Consider four vertices $a,b,c,d$ in $K_{2,n-2}$: $a$ and $b$ have degree $n-2$
while $c$ and $d$ are chosen arbitrarily.
Let:
$a$ buy the edge $ac$ with time label $n-2$;
$b$ buy the edge $bd$ with time label $n-2$;
$c$ buy the edge $bc$ with time label $1$;
$d$ buy the edge $ad$ with time label $1$;
any other vertex $e$ buy both edges $ae$ and $be$ at time labels $i$ and
$n-1-i$, for $2\leq i\leq n-3$.
In order to conclude, it only remains to prove that this strategy profile (which has optimal social cost) is a Nash equilibrium when $n>6$.

To this aim, we analyse all possible deviations. Let $e^+$ and $e^-$ denote the number
of edges a player adds and removes, respectively. In $K_{2,n-2}$, each player
has either $1$ or $2$ edges, any deviation  $e^+-e^-\leq-2$  would isolate the
player, this leaves us with five cases.
\begin{itemize}
\item \textit{Adding at least one edge $(e^+\geq1,e^-=0)$:}
Since the diameter of $K_{2,n-2}$ is $2$, any added edge decreases the communication cost by at most 1. As
$\alpha\geq2$, no player would purchase an additional edge.
\item \textit{Removing or swapping $(e^+\leq1,e^-=1)$ the unique edge of $a, b, c$ or $d$:}
Each of these players owns a single edge connecting them to the
opposite partition. Removing it (even while adding a new edge with minimum
time label, that is $1$, to another player) would disconnect them from one
another. No such deviation is profitable.
\item \textit{Removing one edge $(e^+=0,e^-=1)$ for a player $e$ with two edges:}
Removing either $ea$ or $eb$ increases $e$'s communication cost by at least
$4$, as the vertices previously reached with distance $2$ through the removed
edge become unreachable or require a detour. Since $\alpha\leq4$ no player $e$
benefits from removing one of their two edges.
\item \textit{Swapping one edge for another $(e^+=1,e^-=1)$ for a player $e$ with two edges:}
Replacing either $ea$ or $eb$ (or decreasing the time label of an
edge to minimum, that is $1$) would either keep $e$'s communication cost
unchanged or increase their distance to the players they used to reach using
the removed edge from $2$ to $3$.  Since the construction cost remains
unchanged, no such deviation is profitable.
\item \textit{Swapping both edges for one $(e^+=1,e^-=2)$ for a player $e$:}
Replacing  $ea$ and $eb$ with a single new edge (or removing one of them while
decreasing the time label of the other to minimum, that is $1$) increases
$e$'s communication cost by at least $4$, as $e$ can no longer reach all four
vertices $a, b, c$ and $d$ directly through the two edges and must now take
detours, and by at most $n-3$, as the distance to each of the remaining
vertices also increases from $2$ to $3$. Since $\alpha\leq4$, this deviation
is unprofitable when $n>6$.
\end{itemize}
An example with $n=8$ is shown in the middle part of
Figure~\ref{fig:8vertices}.
\qed
\end{proof}
\begin{observation}
Combining Observation~\ref{obs:K_n}, Theorem~\ref{lemma:basic}, and a few case-analysis for $n\leq6$, it follows for the out-reachability dynamic network creation game with atomic cost $\alpha$ and $n$ players, using strict temporal shortest distance that:
\begin{itemize}
    \item if $\alpha\leq4$, then $PoS(\alpha,n)=1$,
    \item if $n\leq6$, then $PoS(\alpha,n)=1$,
    \item otherwise, $PoS(\alpha,n)\leq\frac{\alpha}{4}+\frac{1}{2}$ and for every $\epsilon>0$ there exists $n_0\in\mathbb N$ such that $n\geq n_0\Rightarrow PoA(\alpha,n)\geq\frac{\alpha}{4}+\frac{1}{2}-\epsilon$.
\end{itemize}
\end{observation}
\begin{proof}
$K_n$ is a Nash equilibrium no matter which value $\alpha$ (Observation~\ref{obs:K_n}).
If $\alpha\leq2$, $K_n$ also has optimal social cost (Theorem~\ref{lemma:basic}).
If $2<\alpha\leq4$, $K_n$ is no more optimal, but $K_{2,n-2}$ is optimal and also a Nash equilibrium (still Theorem~\ref{lemma:basic}).
Hence, $PoS(\alpha,n)=1$ for $\alpha\leq4$.
For the cases when $n\leq6$, an exhaustive case analysis allows to conclude that the price of stability is $1$.

Otherwise, $K_{2,n-2}$ is optimal while $K_n$ is a Nash equilibrium.
The price of stability is less than the ratio of the social cost of $K_n$ divided by that of $K_{2,n-2}$:
\[
PoS(\alpha,n)\leq\frac{2n(n-1)+(\alpha-2)(\frac{n(n-1)}{2})}{2n(n-1)+(\alpha-2)(2n-4)}\leq PoA(\alpha,n).
\]

We now study the function $f$ defined as:
\[
\begin{aligned}
f(\alpha,n)&=\frac{2n(n-1)+(\alpha-2)(\frac{n(n-1)}{2})}{2n(n-1)+(\alpha-2)(2n-4)}\\
&=\frac{1}{4}\times\left(\frac{\alpha n^2+2n^2-\alpha n-2n}{n^2+\alpha n-3n-2\alpha+4}\right)\\
&=\frac{1}{4}\times\left(\alpha+2-\frac{(\alpha^2-4)(n-2)}{(n-2)(n+\alpha-1)+2}\right).
\end{aligned}
\]

When $\alpha>4$ and $n>6$, the term $\frac{(\alpha^2-4)(n-2)}{(n-2)(n+\alpha-1)+2}$ is strictly positive and its limits when $n\to+\infty$ is $0$.
\qed
\end{proof}

\subsection{High cost of communication for big $\alpha$}
\label{sec:shortest}
\begin{theorem}
\label{lemma:diminishedfan_shortest}
Let $\alpha\geq2$ and let $n\geq4$ be an integer.
For the out-reachability dynamic network creation game with atomic cost $\alpha$ and $n$ players, using strict temporal shortest distance, there exists a labelling of a graph called $L_n$ with social cost $c(L_n)\geq\alpha(2n-3)+\frac{1}{48}(n-2)^3$.

Moreover, if $\alpha\geq\frac{1}{8}n^2$ then it is the communication temporal graph of a Nash equilibrium.
\end{theorem}
\begin{proof}
We define the following (simple) labelling of a graph, called $L_n$, together with a strategy profile.
Let $v_0, v_1, v_2, v_3$ be the base vertices in $L_n$, their edges are defined as : $v_0$ buys an edge to $v_1$ and $v_3$ with time label $2$ for both, $v_2$ buys an edge to $v_1$ with time label $1$ and $v_3$ buys an edge to $v_1$ and $v_2$ with time labels $1$ and $3$ respectively.
If $n=5$ then $v_4$ buys one edge to $v_0$ with time label $1$, while $v_3$ buys one extra edge to $v_4$ with time label $4$.

For the case when $n\geq6$, we still have the same configuration as before with the base vertices $v_0, v_1, v_2, v_3$.
Moreover, $v_3$ also buys an edge to $v_5$ with time label $4$.
For $4\leq i\leq n-3$, $v_i$ buys an edge to $v_0$ with time label $1$; moreover, if $i$ is odd, $v_i$ also buys two other edges: one edge to $v_{i-1}$ with time label $i$ and one edge to $v_{i+2}$ with time label $i+1$.
If $n\geq6$ and $n$ is even, then $v_{n-2}$ buys an edge to $v_0$ with time label $1$; while $v_{n-1}$ buys two edges: one to $v_0$ with time label $1$ and one to $v_{n-2}$ with time label $n-1$.
If $n\geq6$ and $n$ is odd, then $v_{n-2}$ buys three edges: one to $v_0$ with time label $1$, one to $v_{n-1}$ with time label $n-1$ and one to $v_{n-3}$ with time label $n-2$; while $v_{n-1}$ buys one edge to $v_0$ with time label $1$.
Note that vertex $v_{n-1}$ has degree $3$ when $n$ is even and $2$ when $n$ is odd.
Examples with $n=8, n=9$, and $n=10$ are given in Figure~\ref{fig:8910vertices}.
\begin{figure}[t]
\centering
~\hfill
\begin{minipage}{0.30\textwidth} 
\centering
\begin{tikzpicture}[
    scale=0.7,
    transform shape,
    vertex/.style={circle, draw, fill=black, inner sep=1.5pt},
    >={Stealth[length=2mm, width=1.5mm]}, 
    every edge/.style={draw, ->},          
    edge label/.style={midway, font=\small, inner sep=2pt}
]

    \node[vertex, label=right:$v_2$] (M1) at (2,  3) {};
    \node[vertex, label=right:$v_3$] (M2) at (2,  2) {};
    \node[vertex, label=right:$v_4$] (M3) at (2,  1) {};
    \node[vertex, label=right:$v_5$] (M4) at (2,  0) {};
    \node[vertex, label=right:$v_6$] (M5) at (2, -1) {};
    \node[vertex, label=right:$v_7$] (M6) at (2, -2) {};

    \node[vertex, label=right:$v_1$] (R1) at (4,  2) {};
    \node[vertex, label=right:$v_0$] (R2) at (4,  1) {};

    \path (R2) edge node[edge label, right] {2} (R1);
    \path (R2) edge node[edge label, above] {2} (M2);
    \path (M1) edge node[edge label, above] {1} (R1);
    \path (M2) edge node[edge label, above] {1} (R1);
    \path (M2) edge node[edge label, left ] {3} (M1);
    
    \path (M2) edge[bend right=45] node[edge label, left] {4} (M4);
    
    \path (M3) edge node[edge label, above] {1} (R2);
    \path (M4) edge node[edge label, above] {1} (R2);
    \path (M4) edge node[edge label, left] {5} (M3);
    
    \path (M4) edge[bend right=45] node[edge label, left] {6} (M6);
    
    \path (M5) edge node[edge label, above] {1} (R2);
    \path (M6) edge node[edge label, above] {1} (R2);
    \path (M6) edge node[edge label, left] {7} (M5);

\end{tikzpicture}
\end{minipage}
\hfill
\begin{minipage}{0.30\textwidth} 
\centering
\begin{tikzpicture}[
    scale=0.7,
    transform shape,
    vertex/.style={circle, draw, fill=black, inner sep=1.5pt},
    >={Stealth[length=2mm, width=1.5mm]}, 
    every edge/.style={draw, ->},          
    edge label/.style={midway, font=\small, inner sep=2pt}
]

    \node[vertex, label=right:$v_2$] (M1) at (2,  3) {};
    \node[vertex, label=right:$v_3$] (M2) at (2,  2) {};
    \node[vertex, label=right:$v_4$] (M3) at (2,  1) {};
    \node[vertex, label=right:$v_5$] (M4) at (2,  0) {};
    \node[vertex, label=right:$v_6$] (M5) at (2, -1) {};
    \node[vertex, label=right:$v_7$] (M6) at (2, -2) {};
    \node[vertex, label=right:$v_8$] (M7) at (2, -3) {};

    \node[vertex, label=right:$v_1$] (R1) at (4,  2) {};
    \node[vertex, label=right:$v_0$] (R2) at (4,  1) {};

    \path (R2) edge node[edge label, right] {2} (R1);
    \path (R2) edge node[edge label, above] {2} (M2);
    \path (M1) edge node[edge label, above] {1} (R1);
    \path (M2) edge node[edge label, above] {1} (R1);
    \path (M2) edge node[edge label, left ] {3} (M1);
    
    \path (M2) edge[bend right=45] node[edge label, left] {4} (M4);
    
    \path (M3) edge node[edge label, above] {1} (R2);
    \path (M4) edge node[edge label, above] {1} (R2);
    \path (M4) edge node[edge label, left] {5} (M3);
    
    \path (M4) edge[bend right=45] node[edge label, left] {6} (M6);
    
    \path (M5) edge node[edge label, above] {1} (R2);
    \path (M6) edge node[edge label, above] {1} (R2);
    \path (M6) edge node[edge label, left] {7} (M5);

    \path (M6) edge[bend right=45] node[edge label, left] {8} (M7);

    \path (M7) edge node[edge label, above] {1} (R2);

\end{tikzpicture}
\end{minipage}
\hfill
\begin{minipage}{0.30\textwidth} 
\centering
\begin{tikzpicture}[
    scale=0.7,
    transform shape,
    vertex/.style={circle, draw, fill=black, inner sep=1.5pt},
    >={Stealth[length=2mm, width=1.5mm]}, 
    every edge/.style={draw, ->},          
    edge label/.style={midway, font=\small, inner sep=2pt}
]

    \node[vertex, label=right:$v_2$] (M1) at (2,  3) {};
    \node[vertex, label=right:$v_3$] (M2) at (2,  2) {};
    \node[vertex, label=right:$v_4$] (M3) at (2,  1) {};
    \node[vertex, label=right:$v_5$] (M4) at (2,  0) {};
    \node[vertex, label=right:$v_6$] (M5) at (2, -1) {};
    \node[vertex, label=right:$v_7$] (M6) at (2, -2) {};
    \node[vertex, label=right:$v_8$] (M7) at (2, -3) {};
    \node[vertex, label=right:$v_9$] (M8) at (2, -4) {};

    \node[vertex, label=right:$v_1$] (R1) at (4,  2) {};
    \node[vertex, label=right:$v_0$] (R2) at (4,  1) {};

    \path (R2) edge node[edge label, right] {2} (R1);
    \path (R2) edge node[edge label, above] {2} (M2);
    \path (M1) edge node[edge label, above] {1} (R1);
    \path (M2) edge node[edge label, above] {1} (R1);
    \path (M2) edge node[edge label, left ] {3} (M1);
    
    \path (M2) edge[bend right=45] node[edge label, left] {4} (M4);
    
    \path (M3) edge node[edge label, above] {1} (R2);
    \path (M4) edge node[edge label, above] {1} (R2);
    \path (M4) edge node[edge label, left] {5} (M3);
    
    \path (M4) edge[bend right=45] node[edge label, left] {6} (M6);
    
    \path (M5) edge node[edge label, above] {1} (R2);
    \path (M6) edge node[edge label, above] {1} (R2);
    \path (M6) edge node[edge label, left] {7} (M5);

    \path (M6) edge[bend right=45] node[edge label, left] {8} (M8);

    \path (M7) edge node[edge label, above] {1} (R2);
    \path (M8) edge node[edge label, above] {1} (R2);
    \path (M8) edge node[edge label, left] {9} (M7);

\end{tikzpicture}
\end{minipage}
\hfill~
\caption{Temporal graphs $L_8$, $L_9$ and $L_{10}$, with strategy profiles.}
\label{fig:8910vertices}
\end{figure}

To prove that the above strategy profile is a Nash equilibrium, we now analyse all possible deviations. Let $e^+$ and $e^-$ denote the number
of edges a player adds and removes, respectively. In $L_n$, apart from $v_0,
v_1, v_2$ and $v_3$, each player has either $1$ or $3$ edges, any deviation
$e^+-e^-\leq-3$  would isolate the player, this leaves us with five cases.
\begin{itemize}
\item\textit{Any strategy modification to players $v_0, v_1, v_2$ or $v_3$:}
We can individually check that these four vertices cannot decrease their cost by any response.
\item \textit{Adding at least one edge $(e^+\geq1,e^-=0)$:}
When $n$ is odd, note that $(v_1,v_3,v_5,\dots,v_{n-2},v_{n-1})$ is a temporal path, whereas
$(v_1,v_3,v_5,\dots,v_{n-1})$ is a temporal path when $n$ is even.
In both cases the shortest temporal distance is $d=\lceil\frac{n-2}{2}\rceil$. 
Here, vertex $v_1$ would be
interested in buying a new edge to $v_{mid}$ where $mid=\lceil\frac{n}{2}\rceil$ if $\lceil\frac{n}{2}\rceil$ is odd and $mid=\lceil\frac{n}{2}\rceil+1$ otherwise.
However, when the atomic
cost is as expensive as proportional to $n^2$, precisely when $\alpha\geq \frac{1}{2}(1+2+\dots+d)\geq\frac{1}{8}n^2$, no vertex would buy extra edges once it is already connected to the rest of the network.
\item \textit{Removing one edge $(e^+=0,e^-=1)$:} Now, none of the vertices $v_0, v_1,
v_2$ and $v_3$ can remove any of their edges, otherwise, they would not be
able to reach every vertex (even if they decrease the time label on another
edge down to the minimum time, that is $1$).
Likewise, all other vertices $v_i$ for $3<i<n-1$, cannot remove the edge to
$v_0$ in order to reach the lower $i$'s. Every vertex $v_i$, for $3\leq i<n-1$
when $i$ is odd, cannot remove the edge to the next $v_{i+2}$ in order to
reach the higher $i$'s as well the edge to $v_{i-1}$ as it is their only path
to it. No player would remove their edges.
\item
\textit{Swapping one edge for another $(e^+=1,e^-=1)$:} For every $i\geq0$,
replacing any edge (or decreasing the time label of an edge to minimum, that
is 1) cannot decrease any player's communication cost: this deviation can
leave the player unable to reach every vertex or increase their distance to
higher $v_i$'s, or at best, leave their communication cost unchanged. Since
the edge cost remains unchanged, no such deviation is profitable.
\item
\textit{Swapping two or three edges for one or two
others$(e^+\geq1,e^-\geq2)$:} For $i>3$ when $i$ is odd, $v_i$ owns three
edges: to $v_0$, $v_{i-1}$ and $v_{i+2}$. The only possible way to swap the
$v_i$'s edge to $v_{i+2}$ while enabling $v_i$ to reach every player is with
an edge to $v_{i+1}$, this increases $v_i$'s distance to all higher $i$'s,
making it unprofitable. Furthermore, there is no way to swap $v_i$'s edge to
$v_{i-1}$ as it is the only way to reach it (even while decreasing the time
label of another edge to minimum, that is $1$). Finally, in order to reach
$v_0$ the edge to $v_0$ can only be replaced by an edge to $v_1$ or $v_2$,
both of these options leave the player $v_i$ unable to reach every vertex.
These three edges are then non-replaceable.
\end{itemize}

For any $n$ (odd or even), there exists in $L_n$ a temporal path of shortest distance at least $d=\lceil\frac{n-2}{2}\rceil$.
Let $(u_0,u_1,\dots,u_{d})$ be one such temporal path.
For $0\leq i<d$, let $k=d-i$.
The sum of shortest temporal distances from $u_i$ to the other vertices is at least the sum of shortest temporal distances from $u_i$ to $u_{i+1}$, from $u_i$ to $u_{i+2}$, $\dots$, from $u_i$ to $u_{d}$, which is exactly equal to $1+2+\dots+(d-i)=\frac{k(k+1)}{2}$.
Hence, the communication cost is higher than $\sum_{k=1}^{d}\frac{k(k+1)}{2}\geq\frac{1}{2}\sum_{k=1}^{d}k^2\geq\frac{1}{6}d^3\geq\frac{1}{48}(n-2)^3$.
\qed
\end{proof}
\begin{observation}
For the out-reachability dynamic network creation game with atomic cost $\alpha$ and $n$ players, using strict temporal shortest distance, a numerical analysis of Theorem~\ref{lemma:diminishedfan_shortest} leads to:
\begin{itemize}
    \item if $n=o(\sqrt\alpha)$, then $\lim\limits_{\alpha\to+\infty}PoS(\alpha,n)=1$,
    \item in particular, if $n=\sqrt{8\alpha}$, then $\lim\limits_{\alpha\to+\infty}PoS(\alpha,\sqrt{8\alpha})\leq5$.
\end{itemize}
\end{observation}
\begin{proof}
We suppose w.l.o.g.\ $\alpha\geq2$.
In this case, the optimal social cost is $c^*=2n(n-1)+(\alpha-2)(2n-4)$, that of temporal graph $K_{2,n-2}$ given in the proof of Theorem~\ref{lemma:basic}.

The graph called $L_n$ in Theorem~\ref{lemma:diminishedfan_shortest} is a Nash equilibrium when $n=o(\sqrt\alpha)$.
It is also a temporally connected graph, implying that the distance between any pair of vertices is less than $n$.
An upper bound for its social cost is then $c(L_n)\leq\alpha(2n-3)+n^3$.
It follows that:
\[
\begin{aligned}
PoS(\alpha,n)
&\leq\frac{\alpha(2n-3)+n^3}{2n(n-1)+(\alpha-2)(2n-4)}\\
&\leq\frac{2\alpha n+n^3}{2n^2+2\alpha n-6n-4\alpha+8}\\
&\leq1+\frac{n^3+4\alpha-2n^2+6n-8}{2\alpha n+2n^2-6n-4\alpha+8}.
\end{aligned}
\]

When $n=o(\sqrt\alpha)$ the dominant term is $2\alpha n$.
Hence, $\lim\limits_{\alpha\to+\infty}PoS(\alpha,n)\leq1$.
Since $PoS(\alpha,n)\geq1$, equality follows.

When $n=\sqrt{8\alpha}$ the value $2n^2-6n-4\alpha+8=12(\sqrt\alpha-\frac{\sqrt2}{2})^2+2$ is positive.
Hence $PoS(\alpha,n)\leq1+\frac{n^3+4\alpha-2n^2+6n-8}{2\alpha n}\leq1+\frac{n^3}{2\alpha n}+\frac{o(\alpha n)}{2\alpha n}$.
\qed
\end{proof}

\section{Best response cycle}
We first formally define a best response using the notation established in Section~\ref{sec:definition}.
Given a strategy profile $\bm{s}$, a strategy $\bm{s}^*(v) \in 2^{V \times \mathbb{N}^*}$ is a \emph{best response} for player $v$ under $\bm{s}$ if it minimizes $v$'s individual cost while keeping all other players' strategies fixed:
\[
c^+[(\bm{s}^*(v),\, \bm{s}_{-v})](v)
\;\leq\;
c^+[(\bm{s}'(v),\, \bm{s}_{-v})](v)
\quad\text{for all } \bm{s}'(v) \in 2^{V \times \mathbb{N}^*},
\]
where $\bm{s}_{-v} = (\bm{s}(u))_{u \neq v}$ denotes the strategies of all players other than $v$.
Note that $\bm{s}$ is a Nash equilibrium if and only if every player's current strategy is already a best response under $\bm{s}$.

A \emph{best response dynamics} is a sequence of strategy profiles $\bm{s}^{(0)}, \bm{s}^{(1)}, \bm{s}^{(2)}, \ldots$ in which, at each step, exactly one player replaces its current strategy with one of its best responses while all others remain unchanged.
A \emph{best response cycle} occurs when such a sequence is periodic: there exists $k \geq 1$ such that $\bm{s}^{(k)} = \bm{s}^{(0)}$.
The existence of a best response cycle is a strong obstruction: it rules out the game being a potential game~\cite{KL13}, meaning that no real-valued function on strategy profiles can decrease monotonically along every sequence of best responses, and hence that best response dynamics need not converge to a Nash equilibrium.

We will present a best response cycle 
of the dynamic network creation game for out-reachability. 
Our construction is an extension of the improving-response-cycle 
in the proof of Theorem~$4$ in the full version of~\cite{BCFGKLS23}. 
\begin{theorem}
\label{theorem:br-cycle}
The dynamic network creation game for out-reachability admits a best response cycle. 
\end{theorem}

\noindent{\bf Sketch of the proof.}
Consider the initial communication temporal graph $\mathcal{G}^{(0)}$ depicted in Figure~\ref{fig:brcycle}, defined over nine vertices $u, v_1, v_2, v_3, v_4, v_5, v_6, x, y$ 
with atomic cost $\alpha \geq 3$. 

We consider a periodic activation schedule $v_1, v_3, v_5, v_1, v_3, v_5, v_1, \ldots$, 
where each player at the three vertices adopts a best response at each activation. 
The players repeatedly remove and add temporal edges to $u$. 
For example, $v_1$'s best response in $\mathcal{G}^{(0)}$ is removal of $(v_1 u, 3)$, 
and one of $v_3$'s best responses in $\mathcal{G}^{(1)}$ is buying $(v_3 u, 3)$.  
After that, $v_5$'s best response in $\mathcal{G}^{(2)}$ is removal of $(v_5 u, 3)$.  
After the six activations, 
the resulting profile $\bm{s}^{(6)}$ is exactly $\bm{s}^{(0)}$, completing the best response cycle of length six. 
\begin{figure}[t]
\centering
\includegraphics[width=0.4\textwidth]{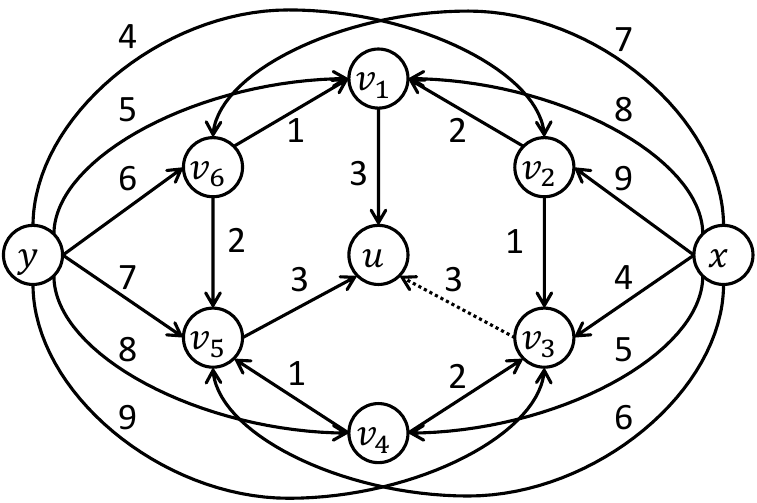}
\caption{Initial communication temporal graph $\mathcal{G}^{(0)}$ for the proposed cycle with $\alpha \geq 3$.
Vertices $v_1$ and $v_5$ hold direct temporal edges to $u$; vertex $v_3$ does not.
The fixed internal edges among $v_1, \ldots, v_6$ include the temporal path $v_1 \to v_6 \to v_5 \to u$ with strictly increasing time labels.}
\label{fig:brcycle}
\end{figure}

Consequently, the dynamic network creation game for out-reachability is not a potential game: best response dynamics can cycle indefinitely, and convergence to neither a Nash equilibrium nor a socially optimal strategy profile is guaranteed.

\section{Conclusion and perspectives}
For temporal graphs with strict temporal shortest distance,
we show a constant bound for the price of anarchy for $\alpha<2$.
For big $\alpha$, we show for $\alpha\geq\frac{1}{8}n^2$ scenarios where the temporal constraints can create communication cost as high as $\frac{1}{48}(n-2)^3$.
This is a very bad situation, hinting that the construction of networks should be guided for a better communication world.
On the other hand, we show for the price of stability that it is at most a constant times $\alpha$, when $n\to+\infty$.
We also establish that the game admits a best response cycle for $\alpha \geq 3$: three players repeatedly toggle their direct edges to a common vertex in a periodic sequence of individually rational moves, returning the game to its initial state without ever reaching a Nash equilibrium.
This demonstrates that the dynamic network creation game is not a potential game, ruling out the existence of a potential function that monotonically decreases along best response dynamics.
An important open question is whether, from any initial strategy profile, there always exists some sequence of best responses that converges to a Nash equilibrium, or whether certain initial configurations are inevitably trapped in a best response cycle regardless of the activation order.

For foremost and fastest temporal distances, we do not yet have corresponding equilibrium lower bounds. 
Preliminary computations suggest that communication costs can be larger than for shortest distance, and determining whether analogous PoA phenomena occur is left as an open question.
In order to better balance this explosion in the communication cost, better modelling works should be done about the optimal construction cost with respect to foremost or fastest distances.
Another modelling direction is a cost-sharing temporal network creation game in which two endpoints that both request the same temporal edge split its maintenance cost. 
This variant may create bilateral-coordination incentives and could have equilibria different from those of the unilateral edge-buying game studied here.

\section*{Acknowledgments}
We are grateful to the anonymous reviewers of all our submissions for their helpful comments which greatly improved the paper.

\bibliography{papers}

\end{document}